\newif\ifready\readyfalse
\pdfoutput=1

\documentclass[acmsmall,nonacm]{acmart}
\usepackage[utf8]{inputenc}
\usepackage{graphicx}
\usepackage{booktabs}
\usepackage{amsthm}
\usepackage{amsmath}
\usepackage{relsize}
\usepackage{algorithmicx}
\usepackage{algorithm}
\usepackage[noend]{algpseudocode}
\usepackage{multicol}
\usepackage[most]{tcolorbox}
\usepackage{color}
\usepackage{pgfplots}
\usepackage{subcaption}
\pgfplotsset{
    compat=1.3,
    legend image code/.code={
        \draw [#1] (0cm,-0.1cm) rectangle (0.6cm,0.1cm);
    },
}
\usepackage{thmtools}
\usepackage{thm-restate}
\usepackage{enumerate}
\usepackage{enumitem}
\setlist{noitemsep,topsep=0pt,parsep=0pt,partopsep=0pt}
\usepackage{mathtools}
\usepackage{xspace}

\usepackage[capitalize,nameinlink]{cleveref}

\theoremstyle{plain}
\newtheorem{theorem}{Theorem}[section]

\newtheorem{lemma}[theorem]{Lemma}

\newtheorem{definition}[theorem]{Definition}

\newcommand{\defn}[1]{\textbf{\textit{#1}}}

\crefname{theorem}{Theorem}{Theorems}
\Crefname{lemma}{Lemma}{Lemmas}
\Crefname{claim}{Claim}{Claims}
\Crefname{observation}{Observation}{Observations}
\Crefname{algorithm}{Algorithm}{Algorithms}
\Crefname{myalgctr}{Algorithm}{Algorithms}
\Crefname{challenge}{Challenge}{Challenges}

\algrenewcommand\algorithmicindent{1em}%

\DeclarePairedDelimiter\ceil{\lceil}{\rceil}

\usepackage{xcolor}
\DeclareMathOperator*{\argmin}{arg\,min}

\DeclareMathOperator{\poly}{poly}

\definecolor{mygreen}{RGB}{20,140,80}
\definecolor{linkcolor}{RGB}{0,0,230}
\definecolor{mylightgray}{RGB}{230,230,230}
\definecolor{verylightgray}{RGB}{245,245,245}

\newcommand{\etal}[0]{et al.\xspace}

\newcounter{myalgctr}

\newtcolorbox{OuterBox}[1][]{%
    breakable,
    enhanced,
    frame hidden,
    interior hidden,
    left=-5pt,
    right=-5pt,
    top=-5pt,
    float=p,
    boxsep=0pt,
    arc=0pt
#1}%

\newtcolorbox{InnerBox}[1][]{%
    enforce breakable,
    enhanced,
    colback=gray,
    colframe=white,
#1}%

\newenvironment{tbox}{
\vspace{0.2cm}
\begin{tcolorbox}[width=\columnwidth,
                  enhanced,
                  boxsep=2pt,
                  left=1pt,
                  right=1pt,
                  top=4pt,
                  boxrule=1pt,
                  arc=0pt,
                  colback=white,
                  colframe=black,
	              breakable
                  ]%
}{
\end{tcolorbox}
}

\newcommand{\tboxhrule}[0]{\vspace{0.1cm} {\color{black} \hrule} \vspace{0.2cm}}

\newenvironment{titledtbox}[1]{\begin{tbox}#1 \tboxhrule}{\end{tbox}}

\newcommand{\eps}{\varepsilon}

\newcommand{\geom}{\mathsf{Geom}}

\newcommand{\fact}{\log^2 n}

\newcommand{\draw}{\geom(\exp(\eps/\fact))}

\algblock{PhaseOne}{EndPhaseOne}
\algblock{PhaseTwo}{EndPhaseTwo}
\algblock{Maintain}{EndMaintain}

\algnewcommand\algorithmicphaseone{\textbf{Phase 1:}}
\algnewcommand\algorithmicphasetwo{\textbf{Phase 2:}}
\algnewcommand\algorithmicmaintain{\textbf{Maintain:}}

\algrenewtext{PhaseOne}[1]{\algorithmicphaseone\ #1}
\algrenewtext{PhaseTwo}[1]{\algorithmicphasetwo\ #1}
\algrenewtext{Maintain}[1]{\algorithmicmaintain\ #1}

\algtext*{EndIf}
\algtext*{EndFor}
\algtext*{EndWhile}
\algtext*{EndPhaseOne}
\algtext*{EndPhaseTwo}
\algtext*{EndMaintain}

\title{Brief Announcement: Improved Massively Parallel Triangle Counting in $O(1)$ Rounds}

\author{Quanquan C. Liu}
\affiliation{%
  \institution{Yale University}
  \city{New Haven}
  \country{USA}}
\email{quanquan.liu@yale.edu}

\author{C. Seshadhri}
\affiliation{%
  \institution{University of California, Santa Cruz}
  \city{Santa Cruz}
  \country{USA}
}
\email{sesh@ucsc.edu}

\date{}

\begin{abstract}
In this short note, we give a novel algorithm for $O(1)$ round triangle counting in bounded arboricity graphs.
Counting triangles in $O(1)$ rounds (exactly) is listed as one of the interesting remaining open problems in the recent 
survey of Im \etal~\cite{im2023massively}.
The previous paper of 
Biswas et al.\ [BELMR20], which achieved the best bounds under this setting, used $O(\log \log n)$ 
rounds in sublinear space per machine and $O(m\alpha)$ total 
space where $\alpha$ is the arboricity of the graph and $n$ and $m$ are the number of vertices 
and edges in the graph, respectively. Our new algorithm is very simple, achieves the optimal $O(1)$ rounds without increasing the 
space per machine and the total space,
and has the potential of being easily implementable in practice.
\end{abstract}

\begin{document}
\renewcommand{\algorithmicrequire}{\textbf{Input:}}
\renewcommand{\algorithmicensure}{\textbf{Output:}}
\algblock{ParFor}{EndParFor}
\algblock{Input}{EndInput}
\algblock{Output}{EndOutput}
\algblock{ReduceAdd}{EndReduceAdd}

\algnewcommand\algorithmicparfor{\textbf{parfor}}
\algnewcommand\algorithmicinput{\textbf{Input:}}
\algnewcommand\algorithmicoutput{\textbf{Output:}}
\algnewcommand\algorithmicreduceadd{\textbf{ReduceAdd}}
\algnewcommand\algorithmicpardo{\textbf{do}}
\algnewcommand\algorithmicendparfor{\textbf{end\ input}}
\algrenewtext{ParFor}[1]{\algorithmicparfor\ #1\ \algorithmicpardo}
\algrenewtext{Input}[1]{\algorithmicinput\ #1}
\algrenewtext{Output}[1]{\algorithmicoutput\ #1}
\algrenewtext{ReduceAdd}[2]{#1 $\leftarrow$ \algorithmicreduceadd(#2)}
\algtext*{EndInput}
\algtext*{EndOutput}
\algtext*{EndIf}
\algtext*{EndFor}
\algtext*{EndWhile}
\algtext*{EndParFor}
\algtext*{EndReduceAdd}

\maketitle
\sloppy

\section{Introduction}
In this short extended abstract, we study the triangle counting problem in the Massively Parallel Computation (MPC) model.
Given a simple, undirected input graph, $G = (V, E)$, the total triangle count is the number of cycles of length three in the graph. This problem
has a wide variety of applications including community detection, spam detection, link recommendation, and social network analysis. 
Triangle counting and enumeration play key roles in database joins and are among the most important problems in database theory.
Because of these applications, we often need to solve the problem on massive graphs and datasets.
For a sample of works on triangle counting on large graphs, both in theory and practice, 
see~\cite{al2018triangle,arifuzzaman2013patric,Aliak,AKK,bar2002reductions,BELMR20,bera2017towards,Dhulipala2021,
ELRS15,ERS20,kane2012counting,kolda2014counting,lai2015scalable, mcgregor2016better,
pagh2012colorful,park2014mapreduce,
ST15,suri2011counting} and references therein. 

The Massively Parallel Computation (MPC) model~\cite{Beame13,GSZ11,karloff2010MapReduce} 
is one of the main models for modeling large distributed systems capable
of processing graphs with billions or even trillions of edges. These systems include
MapReduce \cite{dean2008mapreduce}, Hadoop \cite{white2012hadoop}, Spark \cite{zaharia2010spark} and Dryad \cite{isard2007dryad}.
In the MPC model, we are given $M$ machines each with 
$S$ memory. The initial input is partitioned arbitrarily across the machines.
Computation is performed using a number of synchronous rounds. In each round, first, local computation is 
performed in each machine using the data that is stored in the machine. Then, data is sent between machines synchronously. 
No machine can receive or send more than $S$ data. The complexity measures we care about in this setting are the number of 
rounds of communication, $R$, the space per machine, $S$, and the total memory used by the computation, which is equal to $M \cdot S$.
This model has gained much interest within the distributed and parallel communities with an abundance of research works published
within just the past few years. For a survey of such works, see~\cite{im2023massively} and references therein.

In this paper, we give a novel algorithm for constant-round triangle counting in bounded arboricity graphs. Bounded 
arboricity graphs are graphs whose edges can be decomposed
into a small number of forests. Specifically, a graph has
arboricity $\alpha$ if the set of edges in the graph 
can be decomposed into at most $\alpha$ forests. Various
papers have demonstrated that a large number of real-world graphs have very small arboricity~\cite{dhulipala2017julienne,shin2018patterns}.
The best previous result of Biswas et al.~\cite{BELMR20} for bounded arboricity graphs
achieved $O(\log \log n)$ rounds in sublinear space per machine and $O(m\alpha)$ total space where $\alpha$ is the arboricity of the graph.
Despite the importance of this problem for a wide variety of applications, 
triangle counting has remained elusive in the MPC model with either $O(1)$ round algorithms
suffering from lower bounds on the count to guarantee good estimates or the best algorithms needing $\omega(1)$ rounds for an exact 
count. In this short note, we give a simple algorithm for the bounded arboricity setting where our algorithm returns an
exact count of the number of triangles in $O(1)$ rounds, $O(n^\delta)$ space per machine for any constant $\delta > 0$,
and using near-linear total space when $\alpha$
is small. An added bonus is that our algorithm is very simple, making it practically implementable for 
real-world networks, where $\alpha$ tends to be small.

\section{Preliminaries}

We formally define the MPC model and arboricity in this section.

\begin{definition}[Arboricity]\label{def:arboricity}
    The \defn{arboricity}, $\alpha$, of a graph $G = (V, E)$ is the minimum number of forests needed to decompose the edges of $G$.
\end{definition}

\begin{definition}[Massively Parallel Computation (MPC) Model]\label{def:mpc-model}
    In the massively parallel computation (MPC) model, there are $M$ machines which communicate with each other in synchronous rounds.
    The input graph, $G = (V, E)$, is initially partitioned across the machines arbitrarily. Each machine has $S$ space and performs
    the following (in order) during each round:
    \begin{enumerate}
        \item Each machine performs (unbounded) local computation using data stored within the machine. (Most reasonable algorithms will not
        use a large amount of time.)
        \item At the end of the round, machines exchange messages synchronously to inform the computation for the next round. The total size
        of messages sent or received by a machine is upper bounded by $S$. 
    \end{enumerate}

    We seek to minimize $S$, the number of rounds of communication, and the total space $S \cdot M$.
    There are three domains for the size of $S$: (i) \emph{Sublinear}: $S = n^{\delta}$ for some constant $\delta \in (0, 1)$;
    (ii) \emph{Near-linear}: $S = \Theta(n\poly(\log n))$; 
    (iii) \emph{Superlinear}: $S = n^{1 + \delta}$ for constant $\delta \in (0, 1)$.
\end{definition}

Throughout, we denote the degree of a vertex $v \in V$ by $\deg(v)$.

\section{$O(1)$ Round Exact Triangle Counting}

In this section, we give a very simple algorithm for counting the exact number of triangles in an input graph $G = (V, E)$ in $O(n^{\delta})$ space per machine
for any constant $\delta > 0$. The main idea behind our algorithm consists of enumerating the wedges adjacent to the lower degree endpoint of every edge.
We first show that the total memory necessary to perform such an enumeration is bounded by $O(m\alpha)$ where $\alpha$ is the arboricity of the input graph
$G = (V, E)$. Chiba-Nishizeki~\cite{CN85} showed the following lemma that bounds the sum of the minimum of the degrees of the endpoints of every edge. For completeness, we include the proof of~\cref{lem:bound-min-deg} in~\cref{app:chiba-nishizeki}.

\begin{lemma}[Chiba-Nishizeki Sum of Minimum Degree Endpoints~\cite{CN85}]\label{lem:bound-min-deg}
    Given an input graph $G = (V, E)$ with arboricity $\alpha$, it holds that $\sum_{(u, v) \in E} \min\left(\deg(u), \deg(v)\right) \leq 2m\alpha$.
\end{lemma}

Given~\cref{lem:bound-min-deg}, we give our MPC algorithm in~\cref{alg:triangle-count}. We use a 
number of MPC primitives which are listed below and have been shown~\cite{Goodrich11,BELMR20} 
to take $O(1/\delta)$ rounds (where $\delta> 0$) in $O(n^{\delta})$
space per machine, and $O(|N|)$ total space ($N$ is the input and $|N|$ is its size):

\begin{enumerate}
    \item \textsc{MPC-Sort}$(N)$: sorts a set $N$ of elements,
    \item \textsc{MPC-Count}$(N)$: counts the number of elements in the multiset $N$,
    \item \textsc{MPC-Duplicate}$(N)$: duplicates the elements in $N$,
    \item \textsc{MPC-Filter}$(N, M)$: returns the set of elements that are in \emph{both} $N$ and $M$,
    \item \textsc{MPC-CountDuplicates}$(N, i)$: given a list of tuples $N$, counts the number of copies of each distinct element in the multiset consisting of the element in the $i$-th index of each tuple; returns a sorted list of tuples $(element, count)$ where $element$ is a distinct element and $count$ is
    its count, and
    \item \textsc{MPC-TagCount}$(N, D)$: given a sorted list $N$ of constant-sized tuples and list of tuples $(element, count)$ in $D$ (where every element in any tuple in $N$ is counted in $D$); tag each element in each tuple in $N$ with its count in $D$.
\end{enumerate}

\setlength{\textfloatsep}{5pt}
\begin{algorithm}
\caption{$O(1)$-Round Exact Triangle Counting}\label{alg:triangle-count}
\begin{algorithmic}[1]
\Require{Graph $G = (V, E)$, constant $\delta > 0$.}
\Ensure{An exact count of the number of triangles in $G$ in $O(1)$ rounds, $O(n^{\delta})$ space per machine, and $O(m\alpha)$ total space.}
\State $E' \leftarrow \textsc{MPC-Duplicate}(E)$.\label{line:duplicate-edge}
\State $SortedE' \leftarrow \textsc{MPC-Sort}(E')$ by both endpoints of edges in $E'$.\label{line:sort-edges}
\State $D \leftarrow \textsc{MPC-CountDuplicates}(SortedE', 0)$ returning degree vector of each vertex.\label{line:adjacent}
\State $TaggedSortedE' \leftarrow \textsc{MPC-TagCount}(SortedE', D)$.\label{line:tag-count}
\For{edge $e = (u, v) \in E$}\label{line:for-loop}
    \State Let $w \leftarrow \argmin(\deg(u), \deg(v))$.\label{line:smaller-degree}
    \State Split $SortedE'[w]$ into size $n^{\delta}$ partitions: $P_{e, 1}, \dots, P_{e, \ceil{\deg(w)/n^{\delta}}}$.\label{line:partition}
    \State Send $e$ and each partition $P_{e, i}$ to a separate machine.\label{line:partition-to-machine}
\EndFor
\For{each machine $M$ containing an edge $e$ and partition $P_{e, i}$ of edges}\label{line:for-machine}
    \State Form wedges using $e$ and each edge in $P_{e, i}$.\label{line:form-wedge}
    \For{each wedge $(a, b, c)$}\label{line:for-wedge}
        \State Construct query $((a, c), e, M)$.\label{line:query}
    \EndFor
\EndFor
\State Let $Q$ be the set of all queries.\label{line:all-queries}
\State $T \leftarrow \textsc{MPC-Filter}(Q, E)$.\label{line:filter}
\State Return $(T, |T|/3)$.
\end{algorithmic}
\end{algorithm}

First, we assume that each vertex is assigned a unique index in $[n]$.
Our algorithm first finds the adjacent edges to each vertex by duplicating each edge in~\cref{line:duplicate-edge} and then 
sorting the edges by both endpoints (\cref{line:sort-edges}). In other words, for each edge $e = (u, v)$ and its duplicate $e' = (u, v)$,
we sort $e$ by $u$ and $e'$ by $v$. Then, we count the number of edges adjacent to each vertex in~\cref{line:adjacent} which allows
us to obtain the endpoint with smaller degree for each edge. Let $D$ be a vector 
containing the degree of each vertex. We then use our degrees stored in $D$ to tag
each edge in $SortedE'$ with its degree (\cref{line:tag-count}).
Then, for each edge (\cref{line:for-loop}), we find the endpoint with smaller degree in~\cref{line:smaller-degree} using $TaggedSortedE'$.
For the endpoint $w$ with smaller degree, we split the adjacency list for $w$ into chunks of size $n^{\delta}$ (\cref{line:partition}).
Let these partitions be $P_{e, 1}, \dots, P_{e, \ceil{\deg(w)/n^{\delta}}}$. Then, we send each $e$ and a partition $P_{e, i}$ to a separate 
machine (\cref{line:partition-to-machine}). Partitioning can be done in $O(1/\delta)$ MPC rounds by first counting the 
number of times $c_u$ each vertex $u$ is the lower degree endpoint, then duplicating each $\{u, v\}$ edge $c_u + c_v$ times and tagging
the duplicate with a unique $i \in [c_u]$ or $j \in [c_v]$; 
finally, we sort all (duplicated) edges by their 
smaller degree endpoint and tag, and partition the sorted list.

For each machine which contains an edge $e$ and corresponding partition $P_{e, i}$ (\cref{line:for-machine}), 
we form paths of vertex length $3$ (2 edges each), otherwise known as \defn{wedges}, between $e$ and each edge in $P_{e, i}$ (\cref{line:form-wedge}).
For each wedge $(a, b, c)$ where $b$ is the middle vertex (\cref{line:for-wedge}), the wedge is a triangle if edge $(a, c)$ exists in the graph.
Thus, we form a query for edge $(a, c)$ (\cref{line:query}). We tag the query with the edge $e$ and the machine $M$ to distinguish between 
different queries for the same edge. We define $Q$ to be the set of all queries (\cref{line:all-queries}).
We then determine the set of queries which are existing edges by using the MPC filter primitive (\cref{line:filter}). 
We use the filter to filter all queries $((a, c), e, M)$ where $(a, c)$ is an edge. The 
filter primitive can be implemented by sorting all queries together with the actual edges. Then, if the count of all
sorted elements that contain $(a, c)$ is greater than the number of queries for $(a, c)$, the edge $(a, c)$ exists.
The filter then returns all queries for which edge $(a, c)$ exists. We let this set of returned queries be $T$ (\cref{line:filter}).
From $(a, c)$ and $e$, one can enumerate the triangle. Finally, we return all filtered queries and the triangle count which is equal to the 
number of filtered queries divided by $3$. 

We now prove the number of rounds, space per machine, and total space used by~\cref{alg:triangle-count}.

\begin{theorem}\label{thm:triangle-count}
    \cref{alg:triangle-count} outputs the set of triangles and triangle count using $O(1/\delta)$ MPC rounds, $O(n^{\delta})$ space per machine
    for any constant $\delta > 0$, and $O(m \alpha)$ total space where $\alpha$ is the arboricity of the graph.
\end{theorem}

\begin{proof}
    We first prove that our algorithm enumerates the set of triangles in the input graph. For each edge $e$, 
    we enumerate all wedges formed using $e$ and the adjacency list of the smaller degree endpoint. For each triangle, 
    $(a, b, c)$, three wedges are enumerated, one initiated by each of the edges in the triangle. Then, for each wedge, 
    we check whether the edge that completes the wedge into a triangle exists in the graph. If it exists, then the wedge
    is returned as a query that is a triangle. The triangle can be obtained from $(a, c)$ and $e$, both contained in the query.
    Since each edge of $(a, b, c)$ initiates the formation of exactly one wedge, the set of enumerated queries is exactly 
    three times the number of distinct triangles. 

    We now show that our algorithm takes $O(1/\delta)$ MPC rounds, $O(n^{\delta})$ space per machine for any constant $\delta > 0$,
    and $O(m \alpha)$ total space where $\alpha$ is the arboricity of the graph. All of our primitives satisfy these measures. 
    We call a constant number of primitives; hence the number of rounds necessary to run all of the primitives is $O(1/\delta)$. Then, 
    sending $e$ and the associated partitions each into a separate machine takes one round. Since the adjacency list is partitioned into $n^{\delta}$
    sized chunks, the total space necessary per machine is $O(n^{\delta})$. Finally, $O(n^{\delta})$ queries are formed per machine. And by~\cref{lem:bound-min-deg},
    at most $O(m\alpha)$ total queries are formed. Thus, the total space usage is $O(m\alpha)$, the number of rounds is $O(1/\delta)$, and space per machine 
    is $O(n^{\delta})$ for any constant $\delta > 0$.
\end{proof}

\subsection{Lower Bound on Number of Rounds}

In this section, we show that our number of rounds is optimal. 
There exists a simple lower bound of $\Omega(1/\delta)$ rounds for triangle counting when the initial graph is partitioned
across multiple machines. Suppose we have multiple disjoint subgraphs partitioned across multiple machines. We show that in this 
worst case setting, computing the number of triangles in the input graph requires $\Omega(1/\delta)$ rounds.
This means that our bounds are tight up to constant factors. 

\begin{lemma}
    Counting the number of triangles in an input graph with $m$ edges and $n < m$ vertices requires $\Omega(1/\delta)$ rounds when 
    the space per machine is $n^{\delta}$ and total space is $m\alpha$.
\end{lemma}

\begin{proof}
    Suppose that each of $m/n^{\delta}$ machines contains the edges of a disjoint subgraph with number of triangles unknown to the other machines. 
    Thus, in order to obtain the triangle count, we must aggregate the counts of the triangles on each individual machine to one machine. Since 
    each machine has space $n^{\delta}$, each machine can receive only up to $n^{\delta}$ counts from other machines. The problem of aggregating
    the counts onto one machine then reduces to the problem of constructing constructing a tree with $m/n^{\delta}$ leaves and where every 
    internal node has degree at most $n^{\delta}$. The height of the tree is then the minimum number of rounds required by any triangle counting algorithm.
    Such a tree must have height $\log_{n^{\delta}}(m/n^{\delta}) = \log_{n^{\delta}}(m) - 1 = \Omega(1/\delta)$. 
\end{proof}

\section{Open Questions}
The remaining open question is to show that counting the number of triangles, even when we allow for a $(1+\eps)$-approximation, 
can be done in $O(1)$ rounds in near-linear or sublinear space per machine, 
without any assumptions on the minimum number of triangles that are present in the graph and without sparsity assumptions
for the graph. For very dense graphs, one can potentially use matrix multiplication techniques. Thus, the interesting
setting is when the graph has arboricity $\alpha = \omega(\poly(\log n))$ but has $o(n^2)$ edges.
Another interesting question is to extend our results to other types of subgraphs beyond triangles.

\appendix
\section{Proof of Lemma~\ref{lem:bound-min-deg}}\label{app:chiba-nishizeki}
\begin{proof}
    Consider the arboricity decomposition of graph $G = (V, E)$. Let $F_1, \dots, F_\alpha$ be the $\alpha$ forests of the 
    decomposition. Suppose we pick an arbitrary root for each tree in each forest $F_i$. Then, we orient the edges from the root
    to the children (toward the leaves). We assign each edge to the node that it is oriented towards. Then, each node in forest $F_i$
    has at most one edge assigned to it. We denote the vertex that edge $e$ is directed towards by $to(e)$.
    Then, we can show that the sum of the minimum degrees is as follows:

    \begin{align}
        \sum_{(u, v) \in E} \min(\deg(u), \deg(v)) &\leq \sum_{1 \leq i \leq \alpha}\sum_{e \in F_i} \deg(to(e)) \label{eq:sum_to_e}\\
        &\leq \sum_{1 \leq i \leq \alpha}\sum_{v \in V} \deg(v) \label{eq:sum_deg_v} \\
        &\leq 2m\alpha.
    \end{align}

    \cref{eq:sum_to_e} follows because $\deg(to(e)) \geq \min(\deg(u), \deg(v))$ where $e = \{u, v\}$.
    \cref{eq:sum_deg_v} follows because each vertex has at most one edge associated with it; since each vertex has at most one
    edge associated with it, it holds that $\deg(to(e))$ of vertex $to(e)$ is counted at most once per forest $F_i$.
\end{proof}

\bibliographystyle{alpha}
\bibliography{ref}

\end{document}